\documentclass[prl,aps,superscriptaddress,amsmath,amsthm,amssymb,twocolumn,reprint]{revtex4-2}

\usepackage{braket,graphicx,float,amsthm}
\usepackage{algpseudocode}
\usepackage{tikz}
\usepackage{caption}
\def\md{\mathrm{d}}
\usepackage[colorlinks=true,linkcolor=blue,anchorcolor=red,citecolor=blue, urlcolor=blue]{hyperref}
\newtheorem{lemma}{Lemma}

\begin{document}
\title{Quadratic Quantum Speedup for Finding Independent Set of a Graph}

\author{Xianjue Zhao}
\affiliation{International Center for Quantum Materials, School of Physics, Peking University, Beijing 100871, China}
\author{Peiyun Ge}
\affiliation{State Key Laboratory of Low Dimensional Quantum Physics, Department of Physics, Tsinghua University, Beijing 100084, China}
\author{Li You}
\affiliation{State Key Laboratory of Low Dimensional Quantum Physics, Department of Physics, Tsinghua University, Beijing 100084, China}
\affiliation{Frontier Science Center for Quantum Information, Tsinghua University, Beijing 100084, China}
\affiliation{Beijing Academy of Quantum Information Sciences, Beijing 100193, China}
\affiliation{Hefei National Laboratory, Hefei 230088, China}
\author{Biao Wu}
\email{wubiao@pku.edu.cn}
\affiliation{International Center for Quantum Materials, School of Physics, Peking University, Beijing 100871, China}
\affiliation{Hefei National Laboratory, Hefei 230088, China}
\affiliation{Wilczek Quantum Center, Shanghai Institute for Advanced Studies,
Shanghai 201315, China}

\date{\today}

\begin{abstract}
A quadratic speedup of the quantum adiabatic algorithm (QAA) for finding independent sets (ISs) in a graph is proven analytically. In comparison to the best classical algorithm with \(O(n^2)\) scaling, where $n$ is the number of vertexes, our quantum algorithm achieves a time complexity of \(O(n^2)\) for finding a large IS, which reduces to \(O(n)\) for identifying a size-2 IS. The complexity bounds we obtain are confirmed numerically for a specific case with the \(O(n^2)\) quantum algorithm outperforming the classical greedy algorithm, that also runs in \(O(n^2)\). The definitive analytical and numerical evidence for the quadratic quantum speedup benefited from an analytical framework based on the Magnus expansion in the interaction picture (MEIP), which overcomes the dependence on the ground state degeneracy encountered in conventional energy gap analysis. In addition, our analysis links the performance of QAA to the spectral structure of the median graph, bridging algorithmic complexity, graph theory, and experimentally realizable Rydberg Hamiltonians. The understanding gained provides practical guidance for optimizing near-term Rydberg atom experiments by revealing the significant impact of detuning on blockade violations.

\end{abstract}

\maketitle

{\it Introduction.}---
Shor's quantum algorithm for factoring~\cite{Shor}, discovered over two decades ago, catalyzed the rapid development of quantum computing. Its quantum advantage, rigorously proven, addressed practical problems and dispelled any doubts about quantum supremacy. Since then, substantial efforts have focused on identifying other quantum algorithms with proven advantages for real-world problems. However, few have been found and rigorously demonstrated, notable exceptions include the HHL algorithm for solving linear systems~\cite{hhl} and Grover's search algorithm~\cite{Grover}.

In 2000, Farhi et al. introduced a new paradigm: quantum adiabatic algorithms (QAA). While computationally equivalent to gate-based quantum algorithms, QAA offers distinct advantages. Notably, it encodes problem constraints directly into a Hamiltonian, thus allowing physical intuition from our extensive experience with quantum mechanics to be leveraged for optimization. Moreover, certain QAAs have recently become implementable with the rapid development of quantum computing hardware, providing a promising path to realize meaningful quantum computation tasks. For example, a QAA for solving the NP-hard maximum independent set (MIS) problem has been experimentally demonstrated on Rydberg atom arrays~\cite{Lukin2022, You, Lukin2017, Lukin2023, Lukin2024, Endres, wupra, yucpl, nsr}.

Despite these advances, many QAA variants, such as the quantum approximate optimization (QAOA) and variational quantum algorithms (VQA)~\cite{adiagrover, QAA, QAOA, AQCRMP}, are yet to demonstrate definitive advantage over classical algorithms for practical problems, either analytically or through convincing numerical simulations. The only exception is the QAA for unstructured search, which by now is known to be capable of matching Grover's algorithm~\cite{Grover}, yielding a quadratic speedup over classical algorithms.

The difficulties in analyzing the complexity of QAA stem from two main challenges. First, the energy gap, which governs the time complexity of QAA, is in general notoriously difficult to compute~\cite{trivialbound, JRS, AQCRMP, YouCountGSD}. Second, classical numerical simulations are limited to relatively small systems, typically only a few dozen of qubits, and thus are insufficient to provide clear demonstrations of quantum advantage.

This work analyzes a QAA designed to find independent sets (ISs) in graphs, that was proposed earlier~\cite{wupra, yucpl, nsr}. The Hilbert space in this QAA is naturally divided into two subspaces: the independent sets, which correspond to ground states, and the non-independent sets, which are excited states. Despite the constant energy gap between these two subspaces in the interaction picture, which removes the need for instance-specific minimum-gap estimates, the complexity of the algorithm remains difficult to analyze due to the highly degenerate eigenspaces associated with ISs. This is essentially  the regime where conventional gap analysis fails.

Our work is made possible with the development of the Magnus expansion in the interaction picture (MEIP)~\cite{Magnus}. As we report below, such an alternative analytical framework is sufficiently powerful to facilitate direct control over transition amplitudes. Unlike gap analysis, we are able to rigorously bound leakage into non-independent sets, hence provide an exact method for analyzing the complexity of the algorithm. Applied to the QAA for finding ISs, we prove that the quantum algorithm finds a large independent set in $O(n^{2})$ time. Remarkably, for the specific task of identifying a size-2 independent set, the runtime reduces to $O(n)$, achieving a quadratic speedup over the fastest classical $O(n^{2})$ algorithm. Hence this work provides a rigorous, complexity-theoretic proof of quantum advantage for a combinatorial optimization problem.

Beyond its theoretical significance, our results also have direct implications for quantum annealing experiments. Our protocol is deterministic, requires no oracles, and is naturally adapted to Rydberg atom array platforms~\cite{You}. As we will show below, our analysis also explains recent experimental observations of blockade violations, and with the insights gained from MEIP we now understand that these violations are not only due to limited runtime but also to excessive detuning, with the latter being a significant contributor to leakage, highlighting detuning as a crucial experimental parameter. Therefore, our framework establishes a provable quantum speedup and offers concrete guidance for optimizing near-term Rydberg-atom-based quantum annealing experiments.

{\it The Hamiltonian of the algorithm.}---
We consider the standard Rydberg-atom encoding of a graph with $n$ vertices, the time-dependent Hamiltonian takes the form  
\begin{align}
  \label{Hamiltonian}
  H_0(s) &= \tfrac{1}{2}\Omega(s)T\sum_j \sigma_j^x
  +\left[-\Delta(s)T\sum_j {n}_j 
  + \omega_0 T \sum_{\langle i,j\rangle}{n}_i{n}_j\right] \nonumber \\
  &\equiv H_{1}+H_{2}\,,
\end{align}
where $H_{1}=\frac{\Omega(s)T}{2}\sum_j \sigma_j^x$, $s\equiv t/T\in[0,1]$ is the scaled time, with respect to the total runtime $T$, $\Omega(s)$ the Rabi frequency, and $\Delta(s)$ the detuning. The operator ${n}_j = (\sigma_j^z+1)/2$ projects onto the Rydberg state of site $j$ with $\sigma_j^z$ counting the difference of a single atom in the Rydberg and ground states, while the blockade strength between connected vertices is assumed to be constant $\omega_0$. Throughout we take $\hbar=1$. 

In the above discussed setup, the computational basis $\{\ket{\phi_1},\ldots,\ket{\phi_N}\}$ with $N=2^n$ naturally corresponds to all allowed classical vertex configurations, and the second term of $H_{2}$ counts the number of violated edges. The Rydberg blockade confines subsequent dynamics within the independent-set (IS) subspace. By choosing $\Omega(0)=\Omega(1)=0$, with $\Delta(0)<0$ and $\Delta(1)>0$, the ground state thus is shown to interpolate between the empty set at $s=0$ and the maximum independent set (MIS) at $s=1$. An adiabatic sweep starting from the vacuum state in the ideal case would prepare a MIS, according to the basic principle of quantum annealing on the Rydberg array platform.

{\it Magnus expansion in the interaction picture.}---
To analyze the dynamics rigorously in more details, we introduce the Magnus expansion in the interaction picture (MEIP). This framework transforms the graph Hamiltonian to the interaction picture by setting $U_{I}(s) = \mathcal{T}e^{-i\int_0^s{{H}}_{2}(s')\mathrm{d}s'}$ to arrive at $\ket{\Phi_{I}(s)}=U_{I}^\dagger(s)\ket{\Phi_S(s)}$ defined by $H_{I}=U_{I}^\dagger H_{1} U_{I}$. With $n_m=\braket{\phi_m|\sum_i n_i|\phi_m}$ denoting the number of excited vertices, we can calculate $U_{I}=\sum_{m=1}^{N}e^{-i(E_m s - n_m T\int_0^s\Delta(s)\mathrm{d}s')}\ket{\phi_m}\bra{\phi_m}$, where $E_m=N_e(m)\omega_0T$ and $N_e(m)$ denotes the number of edges in $\ket{\phi_m}$. We note that $U_{I}(s)$ is a diagonal matrix whose elements are complex numbers with modulus 1. Thus, assuming $\Phi_{I}(s)=\sum_{i}c_i'(s)\ket{\phi_i}$ and $\Phi_{S}(s)=\sum_{i}c_i(s)\ket{\phi_i}$, we find $|c_i'(s)|^2=|c_i(s)|^2$, indicating that we can study $H_{I}$ and $\Phi_{I}(s)$ instead when microstate probability distribution is the only concerned quantity. Thus
\begin{align}
  H_{I}=\frac{1}{2}\Omega(s)T\sum_{\substack{m=1,\cdots,N\\ \text{HD}(l,m)=1}}e^{i  T\int_0^s[\omega_{ml}-(n_m-n_l)\Delta(s')]\mathrm{d}s'}|\phi_m\rangle\!\langle\phi_l|\,,
\end{align}
where $\omega_{ml}=\omega_0(N_e(m)-N_e(l))$ and HD denotes the Hamming distance with the related details given in the appendix. It is important to note that the couplings acquire oscillatory phases proportional to the number of violated edges that would be averaged out in a stroboscopic sense and consequently lead to strongly suppressed leakage out of the IS subspace.

In this picture, the system's evolution operator over each oscillation period can be represented by an average Hamiltonian expanded via the Magnus series~\cite{Magnus}, with the expansion separating naturally into two parts: terms that preserve the IS subspace and terms responsible for leakage. By bounding the latter using operator norm estimates, we obtain an explicit upper bound on the probability of leaving the IS subspace as given below,
\begin{align}
\label{bound}
    \sqrt{1-P_{\text{IS}}}
    \leq a_1a_3 n^2\tau + (a_1^2a_3+ \frac{\pi}{8} \Omega_{\text{max}}^3T^3)n^3\tau^2 + O(n^4\tau^3)\,,
\end{align}
where $P_{\text{IS}}$ denotes the probability of
staying in the IS subspace, and $a_1=\Omega_{\text{max}}T/2$, \\$a_3=\left({\Omega_{\text{max}}\Delta_{\text{max}}T^2}+{2\pi KT}\right)/{4\pi}$, $\tau=2\pi/\omega_0T$. $\Delta_{\text{max}}$ and $\Omega_{\text{max}}$ are the maximum values of $|\Delta(s)|$ and $ |\Omega(s)|$ respectively. $\Omega(s)$ is Lipschitz continuous and the Lipschitz constant is bounded by the constant $K$. Importantly, the upper bound of Eq.~\eqref{bound} vanishes asymptotically provided the blockade strength $\omega_0$ grows at least as $O(n^{2})$ and the parameters $\Delta_{\text{max}},\Omega_{\text{max}}$, and $T$ are independent of $n$, or equivalently, if $\omega_0$ is fixed, the runtime $T$ scales as $O(n^{2})$ and $\Delta_{\text{max}}$ as well as $\Omega_{\text{max}}$ both scale as $O(n^{-1})$. The details of this proof can be found in the appendix.

The analysis above thus yields a rigorous guarantee: the leakage into non-IS states can be made arbitrarily small with only a quadratic overhead, affirming that the Rydberg-atom-based quantum annealing finds ISs in $O(n^{2})$ time at constant energy cost per atom. Hence, an IS can always be found, although its precise size remains to be determined.

Before formally discussing the quantum speedup, we first numerically compare the performance of the  $O(n^{2})$ quantum algorithm and the $O(n^{2})$ classical greedy algorithm for finding large ISs. Based on the results shown in Fig. \ref{fig:4}, we find that the quantum algorithm often yields larger ISs. 
\begin{figure}[H]
  \centering
  \includegraphics[width=0.45\textwidth]{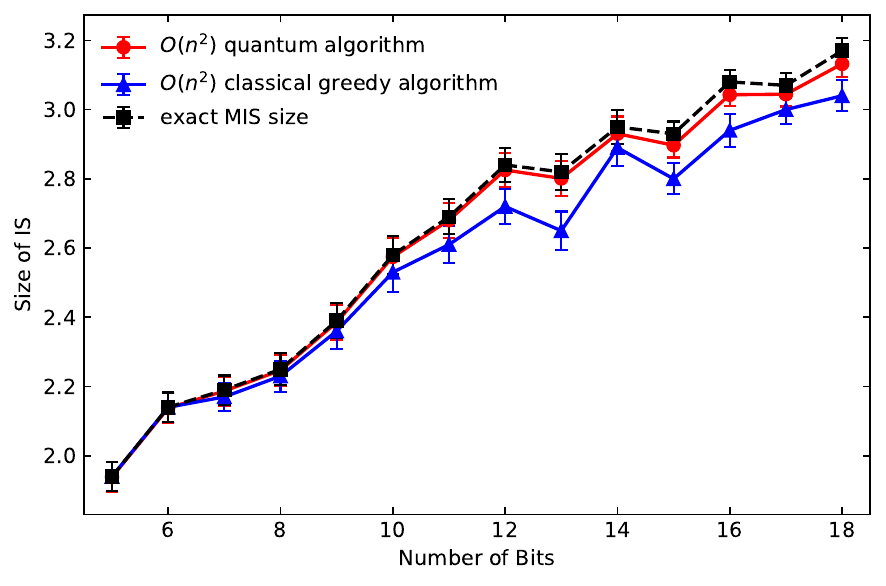}
  \caption{\textbf{$O(n^2)$ Quantum algorithm versus $O(n^2)$ classical algorithm.} The circular data points represent the average size of the IS found by the $O(n^2)$ quantum algorithm. The triangular data points represent the average size of the IS found by the $O(n^2)$ classical greedy algorithm. The square data points represent the exact size of the MIS for the Erd\H{o}s--R\'enyi graph with $p=0.8$ we tested.}
  \label{fig:4}
\end{figure}
For the $O(n^2)$ quantum algorithm, we set $\Omega(s)= \sin(\pi s)$ and $\Delta(s)= \cos(\pi s)$~\cite{wupra,yucpl,nsr} to maximize the sizes of the ISs found. The minimum-degree greedy algorithm is chosen as the classical counterpart for the following reasons: both algorithms share the same $O(n^2)$ complexity; and within this complexity constraint, the minimum-degree greedy algorithm demonstrates outstanding average performance among all classical algorithms~\cite{greedy1,greedy2}.

{\it Quantum speedup.}---
In this section, we provide strict proof for the quadratic quantum speedup of finding IS of size 2. To benchmark the quantum algorithm's performance, we first recall that the best known classical algorithms for finding IS of fixed size $k$. A brute-force search requires $O(n^k k^2)$ time, while more sophisticated methods achieve modest improvements for $k\geq 3$ using fast matrix multiplication~\cite{IS3,IS4567}. However, for the simplest nontrivial case of size-2 IS, the fastest classical runtime remains a $O(n^{2})$ scaling as long as the edge number of the graph is $O(n^{2})$. This makes the problem of finding size-2 IS an ideal case to demonstrate quantum advantage.  
Assuming the quadratic cost in blockade strength is fulfilled so that we can project the Rydberg Hamiltonian onto the IS subspace and obtain an effective Hamiltonian,
\begin{align}
  H_{\text{eff}}=T(\frac{1}{2}\Omega(s) A-\Delta(s) D)\,,
\end{align}
with adjacency matrix $A$ and degree matrix $D$~\cite{nsr}. This Hamiltonian describes 
a quantum walk or diffusion on the  median graph corresponding to the original graph~\cite{nsr,mediangraph}
as shown in Fig. \ref{fig:graph}. Setting $\Delta(s)=0$,~$\Omega(s)=\text{const}$, and starting from the vacuum state (the empty set), the walk rapidly spreads amplitude into the subspace of size-2 ISs. 
\begin{figure}[H]
    \centering
    \begin{tikzpicture}
    \node at (-3,1.5) {(a)};
    \node at (-3,-0.3) {(b)};
    \coordinate (1) at (-2,1.2);
    \coordinate (2) at (0,1.2);
    \coordinate (3) at (2,1.2);
    \coordinate (4) at (-2,0.2);
    \coordinate (5) at (0,0.2);
    \node [draw, fill=black, shape=circle, scale=0.3, label=above:$x_1$] at (1) {0};
    \node [draw, fill=black, shape=circle, scale=0.3, label=above:$x_2$] at (2) {0};
    \node [draw, fill=black, shape=circle, scale=0.3, label=above:$x_3$] at (3) {0};
    \node [draw, fill=black, shape=circle, scale=0.3, label=left:$x_4$] at (4) {0};
    \node [draw, fill=black, shape=circle, scale=0.3, label=right:$x_5$] at (5) {0};
    \draw (1)--(2);
    \draw (2)--(3);
    \draw (1)--(4);
    \draw (4)--(5);
    \draw (4)--(2);
    \draw (5)--(2);

  \coordinate (a) at (0,-0.5);
  \coordinate (b1) at (-3,-2);
  \coordinate (b2) at (-1.5,-2);
  \coordinate (b3) at (0,-2);
  \coordinate (b4) at (1.5,-2);
  \coordinate (b5) at (3,-2);
  \coordinate (c1) at (-2.25,-4);
  \coordinate (c2) at (-0.75,-4);
  \coordinate (c3) at (0.75,-4);
  \coordinate (c4) at (2.25,-4);
  \coordinate (d) at (0,-5.5);
  \draw (a)--(b1);
  \draw (a)--(b2);
  \draw (a)--(b3);
  \draw (a)--(b4);
  \draw (a)--(b5);
  \draw (b1)--(c1);
  \draw (b1)--(c2);
  \draw (b3)--(c1);
  \draw (b3)--(c3);
  \draw (b3)--(c4);
  \draw (b4)--(c3);
  \draw (b5)--(c2);
  \draw (b5)--(c4);
  \draw (d)--(c1);
  \draw (d)--(c2);
  \draw (d)--(c4);
  \node [draw, shape=rectangle,fill=white, scale=1] at (a) {$\{00000\}$};
  \node [draw, shape=rectangle,fill=white, scale=1] at (b1) {$\{10000\}$};
  \node [draw, shape=rectangle,fill=white, scale=1] at (b2) {$\{01000\}$};
  \node [draw, shape=rectangle,fill=white, scale=1] at (b3) {$\{00100\}$};
  \node [draw, shape=rectangle,fill=white, scale=1] at (b4) {$\{00010\}$};
  \node [draw, shape=rectangle,fill=white, scale=1] at (b5) {$\{00001\}$};
  \node [draw, shape=rectangle,fill=white, scale=1] at (c1) {$\{10100\}$};
  \node [draw, shape=rectangle,fill=white, scale=1] at (c2) {$\{10001\}$};
  \node [draw, shape=rectangle,fill=white, scale=1] at (c3) {$\{00110\}$};
  \node [draw, shape=rectangle,fill=white, scale=1] at (c4) {$\{00101\}$};
  \node [draw, shape=rectangle,fill=white, scale=1] at (d) {$\{10101\}$};
    \end{tikzpicture}
    \caption{\textbf{A graph and its corresponding median graph.} (a) An original graph for the independent set problem; (b) the corresponding median graph. Each box (or vertex) in the median graph represents an independent set.}
    \label{fig:graph}
\end{figure}
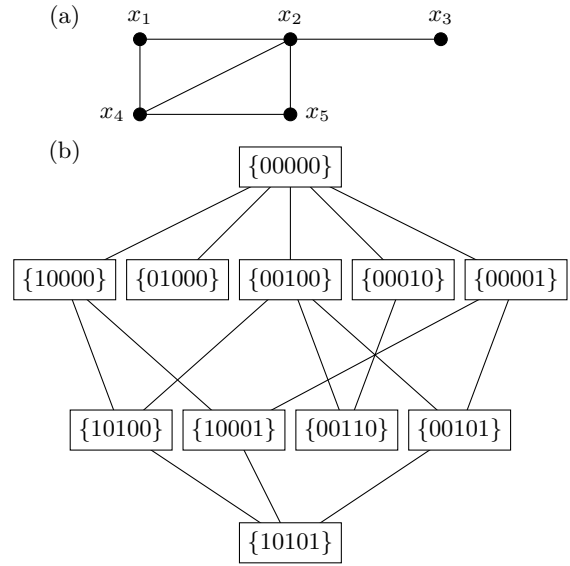

After a very short time $s$, we find
\begin{align}
\label{bound1}
    P_{|\text{IS}|=2}=\frac{1}{8}\Omega^2T^2s^4\left(\frac{n(n-1)}{2}-m\right)^2+O(s^5)\,,
\end{align}
with $n$ number of vertices and $m$ number of edges. For generic graphs which are not too dense, the number of edges in these graphs scales as $\Theta(cn^2),0\leq c<0.5$, so that within an $O(1)$ evolution time the probability of observing a size-2 IS already grows as $\sim n^4 s^4$. The resulting hitting time scales as $s = \kappa n^{-1}$ with constant $\kappa$, well within the available runtime. Taking the limit $n\to \infty$, we obtain the asymptotic expression for the lower bound of the success probability 
\begin{equation}
\label{bound2}
    P_{|\text{IS}|=2}\geq (0.5-c)^2\Omega^2T^2\kappa^4/8\,.
\end{equation}
Fixing the blockade strength to a constant, the above result translates into an overall runtime scaling of $O(n)$, a quadratic improvement over the classical $O(n^2)$ bound. Thus, we provide a rigorous demonstration of quantum speedup for a combinatorial optimization task, and a quantum advantage similar to the well known  Grover-type quadratic improvement.  For extremely dense graphs ($m=n^2/2+o(n^2)$), the quantum speedup becomes slightly slower. Assuming $n(n-1)/2-m= \beta n^{\gamma}+o(n^{\gamma})$ with $0<\gamma<2$ and $\beta>0$, the hitting time of our quantum algorithm becomes $O(n^{2-0.5\gamma})$, which remains faster than the classical $O(n^2)$ runtime.

The above conclusions can be validated on a well-known class of random graphs, the Erd\H{o}s--R\'enyi graph~\cite{erdos}, as shown in Fig. \ref{fig:3}. The simulation results show without doubt that the algorithm's actual performance aligns perfectly with theoretical expectations, i.e., finding the size-2 IS with a probability greater than a known constant within $O(n)$ time.

\begin{figure}[H]
  \centering
  \includegraphics[width=0.45\textwidth]{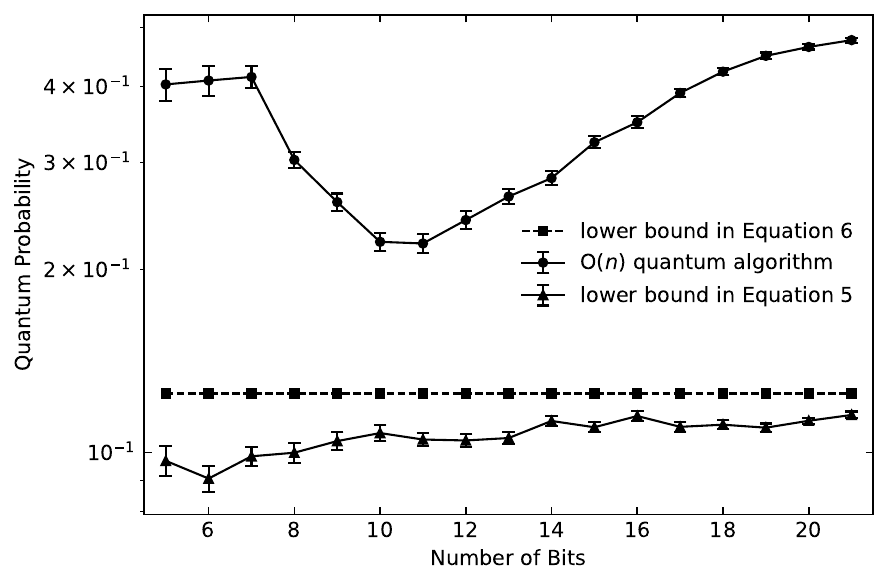}
  \caption{\textbf{The $O(n)$ quantum algorithm's performance for finding size-2 IS.}  The three lines from top to bottom respectively present: the algorithm's actual performance, the asymptotic lower bound of the success rate provided in Eq.~\eqref{bound2}: $(0.5-c)^2\Omega^2T^2\kappa^4/8$, and the theoretical lower bound of the success rate provided in Eq.~\eqref{bound1}: $\Omega^2T^2\kappa^4\left({n(n-1)}/{2}-m\right)^2/8n^4$. We select a class of graphs with appropriate density, namely the Erd\H{o}s--R\'enyi graph with $p=0.5$, to run the quantum algorithm.}
  \label{fig:3}
\end{figure}

It is worth pointing out that the time complexity provided here represents the characteristic time at which the wavefront propagates to the specific sized IS. Longer runtime does not necessarily guarantee a higher algorithm success rate, as subsequent backward scattering will result in coherent superposition of ISs, analogous to the case happened in the Grover algorithm~\cite{Grover}. While some may question the significance of the discovered quantum speedup by arguing that experimentally implementing this Hamiltonian requires $O(n^2)$ time, we believe nevertheless that it is important to distinguish between the time required to prepare the input needed for the algorithm and the runtime of the algorithm itself—the latter determines the algorithm's true time complexity~\cite{AQCRMP}. For classical algorithms, the adjacency matrix of the graph is required as input, whose construction also takes $O(n^2)$ time.

Finally, we would like to readdress the significance of the analytical method we have developed, namely the MEIP. We note that traditional adiabatic conditions~\cite{trivialbound,JRS} are insufficient to capture the behavior we obtain here as energy-gap estimates are either non-rigorous or exponentially pessimistic due to ground-state degeneracy~\cite{Schwinger,Inconsistency,WuComment,Ramsey}. In contrast, MEIP provides a direct and scalable route to bounding leakage and runtime. We expect it will find broad use in certifying quantum advantage especially in regimes where standard gap analysis fails.

{\it Implications for experiment.}---
An important implication of our result is that the quantum annealing protocol is fully deterministic and can be faithfully implemented on Rydberg arrays without relying on oracles. Beyond theoretical significance, our analysis also sheds light on existing experimental simulations. In the landmark 2022 study~\cite{Lukin2022}, a substantial fraction of measured states violated the blockade constraint, failing to form ISs. While such leakage is often attributed solely to insufficient runtime $T$ or blockade strength $\omega_0$, our results given in Eq.~\eqref{bound} reveal that large detuning $\Delta(s)$ also plays a decisive role—even though the detuning term formally commutes with the blockade Hamiltonian. This finding explains naturally why blockade violations persist in experiments despite carefully chosen sweep profiles.  

From a practical perspective, our analysis offers guidance for optimizing experimental parameters: it identifies detuning as a key lever for controlling leakage, complementary to simply extending the runtime. More importantly, our results guarantee theoretically that solving the MIS problem using Rydberg platforms is sufficient to demonstrate a quantum advantage in practically meaningful tasks, and our work demonstrates how to achieve such quantum advantage for this problem. Current experimental limitations mainly stem from finite coherence times and the restricted graph topologies of reconfigurable  two-dimensional atom arrays. Recent experiments have made significant progress for the former~\cite{6100}. The latter can be partially overcome by graph-embedding techniques~\cite{liujinguounitdisk} or by introducing ancillary qubits~\cite{rydbergwire}. These and other related connections highlight how our rigorous complexity results can directly influence the design and interpretation of near-term quantum annealing experiments.

More broadly, our work introduces MEIP as a new framework for certifying the complexity of quantum adiabatic algorithms. By bounding leakage directly in the interaction picture, MEIP bypasses the limitations of conventional gap-based analysis and enables a rigorous proof that Rydberg-based quantum annealers can find independent sets in 
$O(n^2)$ time, and even $O(n)$for size-2 ISs. This establishes a strict quadratic quantum speedup over the best-known classical bound. The mapping of constrained dynamics to a quantum walk on median graphs further reveals how graph-theoretic structure controls algorithmic performance, linking complexity theory, quantum dynamics, and experimental feasibility in a unified perspective. These results establish that solving MIS on programmable Rydberg hardware is already sufficient to demonstrate meaningful quantum advantage, while positioning MEIP as a broadly applicable analytical tool for the next generation of provably efficient quantum algorithms.

{\it Acknowledgments.}---X.Z. would like to thank Hongye Yu and Zhigang Hu for helpful discussions. XJZ and BW are supported by National Natural Science Foundation of China (NSFC) (Grants No.  92365202, No. 12475011 and No.  11921005), the National Key R\&D Program of China (2024YFA1409002), the Shanghai Municipal Science and Technology Project (25LZ2601100) and by Shanghai Municipal Science and Technology Major Project (Grant No.2019SHZDZX01). PYG and LY are supported by NSFC (Grants No. 12361131576 and No. 92265205), and by the Innovation Program for Quantum Science and Technology (2021ZD0302100).

\bibliography{ref}
\bibliographystyle{apsrev4-2}

\newpage

\appendix

\section{The details of MEIP}
Let's first recall the Rydberg atom Hamiltonian, 
\begin{align}
  H_0(s) &= \tfrac{1}{2}\Omega(s)T\sum_j \sigma_j^x
  +\left[-\Delta(s)T\sum_j {n}_j 
  + \omega_0 T \sum_{\langle i,j\rangle}{n}_i{n}_j\right] \nonumber \\
  &\equiv H_{1}+H_{2}\,,
\end{align}
where $s\equiv t/T\in[0,1]$ is the normalized time, $T$ the total runtime, $\Omega(s)$ the Rabi frequency, and $\Delta(s)$ the detuning. The operator ${n}_j = (\sigma_j^z+1)/2$ projects onto the Rydberg state of site $j$, while the interaction $\omega_0$ encodes the blockade between connected vertices. Throughout we take $\hbar=1$. The computational basis $\{\ket{\phi_1},\ldots,\ket{\phi_N}\}$ with $N=2^n$ naturally corresponds to classical vertex configurations and $N_e(m)$ denotes the number of edges in $\ket{\phi_m}$. Here we directly apply MEIP to the Rydberg atom Hamiltonian. To transform into the interaction picture, we set 
\begin{equation}
    U_{I}(s) = \mathcal{T}e^{-i\int_0^s{{H}}_{2}(s')\mathrm{d}s'}\,,
\end{equation}
with $\ket{\Phi_{I}(s)}=U_{I}^\dagger(s)\ket{\Phi_S(s)}$, which bring into the interaction picture: $H_{I}=U_{I}^\dagger H_{1} U_{I}$. With $n_m=\braket{\phi_m|\sum_i n_i|\phi_m}$ denoting the number of excited vertices, we can calculate 
\begin{equation}
    U_{I}=\sum_{m=1}^{N}e^{-i(E_m s - n_m T\int_0^s\Delta(s)\mathrm{d}s')}\ket{\phi_m}\bra{\phi_m}\,,
\end{equation}
where $E_m=N_e(m)\omega_0T$. We note that $U_{I}(s)$ is a diagonal matrix whose elements are complex numbers with modulus 1. Thus, assuming 
\begin{equation}
    \Phi_{I}(s)=\sum_{i}c_i'(s)\ket{\phi_i}\,,
\end{equation}
and
\begin{equation}
    \Phi_{S}(s)=\sum_{i}c_i(s)\ket{\phi_i}\,,
\end{equation}
it is clear that $|c_i'(s)|^2=|c_i(s)|^2$, implicating that we can simply study $H_{I}$ and $\Phi_{I}(s)$ when probability distribution is the only concerned quantity. Note that
\begin{equation}
    H_{1}=\frac{\Omega(s)T}{2}\sum_j \sigma_j^x\,,
\end{equation}
and thus
\begin{align}
  H_{I}=\frac{1}{2}\Omega(s)T\sum_{\substack{m=1,\cdots,N\\ \text{HD}(l,m)=1}}e^{i  T\int_0^s[\omega_{ml}-(n_m-n_l)\Delta(s')]\mathrm{d}s'}|\phi_m\rangle\!\langle\phi_l|\,,
\end{align}
can be calculated directly, where 
\begin{equation}
    \omega_{ml}=\omega_0(N_e(m)-N_e(l))\,,
\end{equation}
and HD denotes the Hamming distance.

To find the leakage out of the IS subspace, we consider the expansion coefficients $c_{ml}$ with $\ket{\phi_l}$ being an IS and $\ket{\phi_m}$ being a non-IS. In this case, $\omega_{ml}\geq \omega_0$ because there is no edge in the IS and at least one edge in the non-IS. We also have $n_m-n_l=1$ because $\text{HD}(l,m)=1$. Thus the coefficients between the IS and 
the non-IS become
\begin{equation}
  \label{hopping}
  c_{ml}=\frac{1}{2}\Omega(s)T\exp\left({i  T\int_0^s[N_e(m)\omega_0-\Delta(s')]\mathrm{d}s'}\right)\,,
\end{equation}
where $N_e(m)\geq1$ is the number of edges in the non-IS $\ket{\phi_m}$. According to the secular approximation, $c_{ml}$ will tend to 0 when $|\Delta|<\omega_0$ , $T(N_e(m)\omega_0-\Delta)\gg 1$, and $\Omega(s)$ changes slowly. Consequently, the IS subspace is decoupled from the rest part of the Hilbert space and the evolution is restricted in the IS subspace. In the following we consider the worst case scenario of $N_e(m)=1$.

To perform the Magnus expansion within each oscillation interval, we assume $T$ is constant and $\omega_0 T/2\pi = L\in \mathbb{N} $, the time domain $[0,1]$ can be divided equally into $L$ intervals by $s_0<s_1<\cdots<s_L$, $s_j = j/L$. The length of each interval $\tau=1/L=2\pi/\omega_0 T$ is a dimensionless small quantity. Assuming $\Delta(s)$ and $\Omega(s)$ are bounded by constant $\Delta_{\text{max}}$ and $\Omega_{\text{max}}$ respectively, and both of which are independent of the graph size $n$; $\Omega(s)$ is Lipschitz continuous and the Lipschitz constant is bounded by a constant $K$ independent of $n$; $\Delta(s)$ is constant in each interval, namely $\Delta(s)=\Delta_j$ for $s\in(s_{j-1},s_j)$.  

During each interval we can define a time-independent average Hamiltonian $\bar{H}_{I}$ through 
\begin{equation}
    U(s_{j-1},s_j)=e^{-i\bar{H}_{I}\tau}\,,
\end{equation}
where $U$ is the propagator of $H_{I}$ and we have omitted the subscript $j$ in $\bar{H}_{I}$ for brevity. This average Hamiltonian can be expanded by the Magnus expansion~\cite{Magnus}:
\begin{equation}
    \bar{H}_{I}=\bar{H}_{I}^{(1)}+\bar{H}_{I}^{(2)}+\bar{H}_{I}^{(3)}+\cdots\,.
\end{equation}
A very sufficient condition for the Magnus expansion to converge~\cite{MagnusConverge} is for any $s\in(s_{j-1},s_{j})$,
\begin{equation}
    \lVert H_{I}(s) \rVert \tau < 1\,,
\end{equation}
where the matrix norm takes the operator norm, namely the spectral radius of $H_{I}(s)$. It is shown by the subsequent lemma that 
\begin{equation}
    \lVert H_{I}(s) \rVert<|\Omega(s)|Tn/2\,.
\end{equation}
Thus, a sufficient convergence condition is 
\begin{equation}
   \pi \Omega_{\text{max}}n/\omega_0<1\,, 
\end{equation}
one simple option is allowing $\omega_0$ to vary with $n$, e.g., $\omega_0=O(n^{\alpha}),\alpha \geq 1$, or equivalently $ \tau=O(n^{-\alpha}),\alpha \geq 1$. Thus, we can treat $\tau$ as a small dimensionless quantity and expand the subsequent results according to its order. 

We divide the whole Hilbert space into subspaces of IS and non-IS. For each term in the Magnus expansion, we can always uniquely divide it into block matrix, $\bar{H}_{I}^{(i)}=\bar{H}^{(i)}_{\text{noleak}}+\bar{H}^{(i)}_{\text{leak}}$, which means if $\ket{\psi}\in$ IS, then $\bar{H}^{(i)}_{\text{noleak}}\ket{\psi}\in$ IS, $\bar{H}^{(i)}_{\text{leak}}\ket{\psi}\in$ non-IS or 0, and analogously if $\ket{\psi}\in$ non-IS, then
$\bar{H}^{(i)}_{\text{noleak}}\ket{\psi}\in$ non-IS or 0, $\bar{H}^{(i)}_{\text{leak}}\ket{\psi}\in$ IS or 0. For each interval, we can estimate the propagator by
  \begin{align}
    U(s_{j-1},s_j)\ket{\psi}&=e^{-i\bar{H}_{I}\tau}\ket{\psi}\nonumber \\
    &=e^{-i(\bar{H}_{I}^{(1)}+\bar{H}_{I}^{(2)})\tau}\ket{\psi}+\ket{\delta\psi_1}\nonumber \\
    &=e^{-i(\bar{H}_{\text{noleak}}^{(1)}+\bar{H}_{\text{noleak}}^{(2)})\tau}\ket{\psi}+\ket{\delta\psi_1}+\ket{\delta\psi_2}\,.
  \end{align}
where $\ket{\delta\psi_1}$ is due to the truncation of the Magnus expansion, $\ket{\delta\psi_2}$ is due to ignoring the $\bar{H}_{\text{leak}}^{(1)}+\bar{H}_{\text{leak}}^{(2)}$. After we done the estimation of $\ket{\delta\psi_1}$ and $\ket{\delta\psi_2}$, the total leakage is bounded by 
\begin{equation}
  \lVert\ket{\psi_{\text{leak}}}\rVert \leq \tau^{-1}(\lVert\ket{\delta\psi_1}\rVert+\lVert\ket{\delta\psi_2}\rVert)\,.
\end{equation}

Let us introduce a linear algebra lemma and some results about the Magnus expansion before the estimation of $\lVert\ket{\delta\psi_1}\rVert$ and $\lVert\ket{\delta\psi_2}\rVert$.

\begin{lemma}\label{lem1}
  If matrix $H$ is a $N\times N$ Hermitian matrix and it has at most $m$ non-zero elements in each row and $|H_{ij}|\leq w, i,j=1,2,\cdots,N$. Then the spectral radius $\rho$ of $H$ satisfies $\rho \leq mw$.
\end{lemma}
  
\begin{proof}
Assuming $|\lambda|=\rho$ and $H v=\lambda v$,$v=(v_1,v_2,\cdots,v_N)$. Assuming $\max(|v_i|)=|v_k|$, then we have:
\begin{align}
|\lambda|&=\frac{\left|\sum_{j=1}^{N}H_{kj}v_j\right|}{|v_k|}\nonumber\\
&\leq\frac{mw|v_k|}{|v_k|}\nonumber\\
&=mw\,.
\end{align}
\end{proof}
  The first three orders of the Magnus expansion are given by
  \begin{equation}
    \bar{H}_{I}^{(1)}=\frac{1}{\tau}\int_{s_{j-1}}^{s_{j}}\md s ~ H_{I}(s)\,,
  \end{equation}
  \begin{equation}
    \bar{H}_{I}^{(2)}=\frac{1}{2i\tau}\int_{s_{j-1}}^{s_{j}} 
    \md s \int_{s_{j-1}}^{s}\md s' ~ [H_{I}(s),H_{I}(s')]\,,
  \end{equation}
  \begin{align}
    \bar{H}_{I}^{(3)}&=-\frac{1}{6\tau}\int_{s_{j-1}}^{s_{j}} 
    \md s \int_{s_{j-1}}^{s}\md s'\int_{s_{j-1}}^{s'}\md s'' \nonumber\\
    &\{\left[H_{I}(s),[H_{I}(s'),H_{I}(s'')]\right] \nonumber \\
    &+\left[[H_{I}(s),H_{I}(s')],H_{I}(s'')\right] \}\,.
  \end{align}

  According to Lemma 1, it is clear that 
  \begin{equation}
    \lVert H^{(1)}_{I}\tau\rVert\leq \frac{1}{2}\Omega_{\text{max}}Tn\tau, \lVert \bar{H}^{(1)}_{\text{noleak}}\tau\rVert\leq \frac{1}{2}\Omega_{\text{max}}Tn\tau\,,
  \end{equation}

  \begin{equation}
    \lVert H^{(2)}_{I}\tau\rVert\leq \frac{1}{8}\Omega_{\text{max}}^2T^2n^2\tau^2, \lVert \bar{H}^{(2)}_{\text{noleak}}\tau\rVert\leq \frac{1}{8}\Omega_{\text{max}}^2T^2n^2\tau^2\,.
  \end{equation}

  There is also a useful bound~\cite{MagnusNorm} for $\lVert \bar{H}_{I}^{(k)} \tau\rVert$:
  \begin{equation}
    \lVert \bar{H}_{I}^{(k)} \tau \rVert \leq \pi\left(\frac{1}{\xi}\int_{s_{j-1}}^{s_j}\lVert H_{I}(s)\rVert\md s \right)^k \,,
  \end{equation}
  where $\xi = \int_{0}^{2\pi}\md x/(4+x(1-\cot(x/2)))\approx 1.0868687$. Then we obtain
  \begin{align}
    \lVert \bar{H}_{I}^{(k)} \tau \rVert  \leq \pi \tau^k \left(\frac{\Omega_{\text{max}}Tn}{2}\right)^k =\pi \left(\frac{\Omega_{\text{max}}T}{2}\right)^k(n\tau )^k \,,
  \end{align}
  \begin{equation}
    \sum_{k=3}^{\infty}\lVert \bar{H}_{I}^{(k)} \tau \rVert<\frac{\pi}{8}(\Omega_{\text{max}}Tn\tau)^3+\frac{\pi}{16}(\Omega_{\text{max}}Tn\tau)^4+O(n^5\tau^5)\,.
  \end{equation}

To estimate $\lVert\ket{\delta\psi_1}\rVert$, we apply the error bound of Suzuki-Trotter decomposition. Define 
\begin{equation}
  \tilde{U}(s_{j-1},s_j)=e^{-i(\bar{H}_{I}^{(1)}+\bar{H}_{I}^{(2)} )\tau}e^{-i\sum_{k=3}^{\infty}\bar{H}_{I}^{(k)} \tau}\,.
\end{equation}
It can be shown that~\cite{trotter}
\begin{align}
  &\lVert {U}(s_{j-1},s_j)-\tilde{U}(s_{j-1},s_j) \rVert\leq\frac{\tau^2}{2}\lVert[\bar{H}_{I}^{(1)}+\bar{H}_{I}^{(2)},\sum_{k=3}^{\infty}\bar{H}_{I}^{(k)}]\rVert \nonumber \\
&\leq\tau^2 (\lVert\bar{H}_{I}^{(1)}\rVert + \lVert\bar{H}_{I}^{(2)}\rVert)\cdot \sum_{k=3}^{\infty}\lVert\bar{H}_{I}^{(k)}\rVert \nonumber \\
&\leq \left(\frac{\Omega_{\text{max}}Tn\tau}{2} + \frac{\Omega_{\text{max}}^2T^2n^2\tau^2}{8}\right)  \left(\frac{\pi\Omega_{\text{max}}^3T^3}{8}(n\tau)^3+O(n^4\tau^4)\right)\nonumber \\
&=\frac{\pi}{16}(\Omega_{\text{max}}Tn\tau)^4+O(n^5\tau^5)\,.
\end{align}
And we can estimate $\lVert e^{-i\sum_{k=3}^{\infty}\bar{H}_{I}^{(k)} \tau} -I \rVert$ by
\begin{align}
  &\lVert e^{-i\sum_{k=3}^{\infty}\bar{H}_{I}^{(k)} \tau} -I \rVert \leq\sum_{k=3}^{\infty}\lVert \bar{H}_{I}^{(k)} \tau \rVert + O\left(\left(\sum_{k=3}^{\infty}\lVert \bar{H}_{I}^{(k)} \tau \rVert\right)^2\right)\nonumber \\
  & \leq \frac{\pi}{8}(\Omega_{\text{max}}Tn\tau)^3+\frac{\pi}{16}(\Omega_{\text{max}}Tn\tau)^4+O(n^5\tau^5)\,.
\end{align}
Thus $\lVert\ket{\delta\psi_1}\rVert$ could be bounded by
\begin{align}
  \lVert\ket{\delta\psi_1}\rVert &\leq \lVert {U}-\tilde{U} \rVert + \lVert e^{-i\sum_{k=3}^{\infty}\bar{H}_{I}^{(k)} \tau} -I \rVert \nonumber \\
  &\leq \frac{\pi}{8}(\Omega_{\text{max}}T n\tau)^3+\frac{\pi}{8}(\Omega_{\text{max}}T n\tau)^4+O(n^5\tau^5)\,.
\end{align}

\newpage

To estimate $\lVert\ket{\delta\psi_2}\rVert$, we should first carefully estimate $\lVert \bar{H}^{(1)}_{\text{leak}}\rVert$ and $\lVert \bar{H}^{(2)}_{\text{leak}}\rVert$. Recall that $\bar{H}_{I}^{(i)}=\bar{H}^{(i)}_{\text{noleak}}+\bar{H}^{(i)}_{\text{leak}}$. Let's first estimate $\lVert \bar{H}^{(1)}_{\text{leak}}\rVert$,  which could be done by examining $H^{(1)}_{ml}$, the coefficient of $\ket{\phi_m}\bra{\phi_l}$ in $\bar{H}_{I}^{(1)}$ with $\ket{\phi_l}$ being an IS and $\ket{\phi_m}$ being a non-IS:
\begin{align}
  &|H^{(1)}_{ml}|=  \frac{1}{\tau}\left|\int_{s_{j-1}}^{s_{j}}\md s ~ c_{ml}(s)\right| \nonumber \\
  &=\frac{T}{2\tau}\left|\int_{s_{j-1}}^{s_{j}}\md s\Omega(s)e^{i T(\omega_0-\Delta_j)(s-s_{j-1})}\right| \nonumber \\
  &=\frac{T}{2\tau}\left|\int_{s_{j-1}}^{s_{j}}\md s[\Omega(s_{j-1})+\Omega(s)-\Omega(s_{j-1})]e^{i T(\omega_0-\Delta_j)(s-s_{j-1})} \right|\nonumber \\
  &\leq \frac{T}{2\tau}\cdot\left( \frac{\Omega(s_{j-1})|e^{i T(\omega_0-\Delta_j)\tau}-1|}{T(\omega_0-\Delta_j)} + K \tau^2 \right)\nonumber \\
  &=\frac{T}{2\tau}\cdot\left( \frac{\Omega(s_{j-1})|e^{-i T\Delta_j\tau}-1|}{T(\omega_0-\Delta_j)} + K \tau^2 \right) \nonumber \\
  &=\frac{T}{2\tau}\cdot\left( \frac{\Omega(s_{j-1})(\Delta_jT\tau+\Delta_j^2T^2\tau^2/2+O(\tau^3))}{T(\omega_0-\Delta_j)} + K \tau^2 \right)\,,
\end{align}
where we already used the Lipschitz continuity of $\Omega(s)$. Thus we obtain
\begin{align}
  &\lVert \bar{H}^{(1)}_{\text{leak}}\tau\rVert \leq n\tau |H^{(1)}_{ml}| \nonumber \\
  &\leq \frac{n}{2}\cdot\left( \frac{\Omega(s_{j-1})(\Delta_jT\tau+\Delta_j^2T^2\tau^2/2+O(\tau^3))}{\omega_0-\Delta_j} + KT \tau^2\right)\nonumber \\
  &\leq \frac{1}{4\pi}\left({\Omega_{\text{max}}\Delta_{\text{max}}T^2}+{2\pi KT}\right)n\tau^2 + \frac{\pi+1}{8\pi^2}\Omega_{\text{max}}\Delta_{\text{max}}^2T^3n\tau^3+O(n\tau^4)\,.
\end{align}
Then we estimate $\lVert \bar{H}^{(2)}_{\text{leak}}\rVert$, which also could be done by examining $H^{(2)}_{ml}$, the coefficient of $\ket{\phi_m}\bra{\phi_l}$ in $\bar{H}_{I}^{(2)}$ with $\ket{\phi_l}$ being an IS and $\ket{\phi_m}$ being a non-IS. Recall the definition of $\bar{H}_{I}^{(2)}$
\begin{equation}
    \bar{H}_{I}^{(2)}=\frac{1}{2i\tau}\int_{s_{j-1}}^{s_{j}} 
    \md s \int_{s_{j-1}}^{s}\md s' ~ [H_{I}(s),H_{I}(s')]\,,
\end{equation}
\begin{align}
  |H^{(2)}_{ml}| &\leq \frac{1}{2\tau} \sum_k \left|\int_{s_{j-1}}^{s_{j}}\md s \int_{s_{j-1}}^{s}\md s' ~ c_{mk}(s)c_{kl}(s')\right| \nonumber \\
  &+ \left|\int_{s_{j-1}}^{s_{j}}\md s \int_{s_{j-1}}^{s}\md s' ~ c_{mk}(s')c_{kl}(s)\right| \nonumber \\
  & \leq \frac{1}{\tau} \sum_k \left|\int_{s_{j-1}}^{s_{j}}\md s  ~ c_{mk}(s)\right|\left| \int_{s_{j-1}}^{s_j}\md s'~c_{kl}(s')\right|\,,
\end{align}

Note that the Hamming distance between $\ket{\phi_l}$ and $\ket{\phi_m}$ is 2. The summation only run through 2 possible configurations and we have 
\begin{align}
  |H^{(2)}_{ml}| &\leq \frac{2}{\tau}\left( \frac{1}{2}\Omega_{\text{max}} T \tau  \cdot |H^{(1)}_{ml}|\right) \nonumber \\
  &= \Omega_{\text{max}} T |H^{(1)}_{ml}|\,,
\end{align}

\begin{align}
  \lVert \bar{H}^{(2)}_{\text{leak}}\tau\rVert &\leq \frac{n(n-1)}{2}\tau |H^{(2)}_{ml}| \nonumber \\
  &\leq \frac{n^2\tau}{2} \Omega_{\text{max}} T |H^{(1)}_{ml}|\,.
\end{align}

\newpage
\begin{widetext}
Again, with the Suzuki-Trotter decomposition, we have
\begin{align}
  &\lVert e^{-i(\bar{H}_{I}^{(1)}+\bar{H}_{I}^{(2)}) \tau}-e^{-i(\bar{H}_{\text{noleak}}^{(1)}+\bar{H}_{\text{noleak}}^{(2)})\tau} e^{-i(\bar{H}_{\text{leak}}^{(1)}+\bar{H}_{\text{leak}}^{(2)})\tau} \rVert \nonumber \\
  &\leq \frac{\tau^2}{2}\lVert[\bar{H}_{\text{noleak}}^{(1)}+\bar{H}_{\text{noleak}}^{(2)},\bar{H}_{\text{leak}}^{(1)}+\bar{H}_{\text{leak}}^{(2)}]\rVert\nonumber \\
  & \leq \lVert\bar{H}_{\text{noleak}}^{(1)}\tau+\bar{H}_{\text{noleak}}^{(2)}\tau\rVert \cdot \lVert \bar{H}_{\text{leak}}^{(1)}\tau+\bar{H}_{\text{leak}}^{(2)}\tau\rVert \nonumber \\
  & \leq \left( \frac{1}{2}\Omega_{\text{max}}Tn\tau+\frac{1}{8}\Omega_{\text{max}}^2T^2n^2\tau^2 \right) \nonumber \\
  & \cdot (1+\frac{n\Omega_{\text{max}}T}{2})\left( \frac{1}{4\pi}\left({\Omega_{\text{max}}\Delta_{\text{max}}T^2}+{2\pi KT}\right)n\tau^2 + \frac{\pi+1}{8\pi^2}\Omega_{\text{max}}\Delta_{\text{max}}^2T^3n\tau^3 +  O(n\tau^4)\right)  \nonumber\\
  &=a_1^2a_3n^3\tau^3 +a_1a_2a_3n^4\tau^4 + a_1a_2a_4n^4\tau^5+(a_2a_3+a_1^2a_4)n^3\tau^4+a_1a_3n^2\tau^3+a_1a_4n^2\tau^4+O(n^3\tau^5)\,.
\end{align}

with $a_1=\frac{1}{2}\Omega_{\text{max}}T$, $a_2=\frac{1}{8}\Omega_{\text{max}}^2T^2$, $a_3=\frac{1}{4\pi}\left({\Omega_{\text{max}}\Delta_{\text{max}}T^2}+{2\pi KT}\right)$, $a_4=\frac{\pi+1}{8\pi^2}\Omega_{\text{max}}\Delta_{\text{max}}^2T^3$.

We can estimate $\lVert e^{-i(\bar{H}_{\text{leak}}^{(1)}+\bar{H}_{\text{leak}}^{(2)})\tau} -I \rVert$ by
\begin{align}
  \lVert e^{-i(\bar{H}_{\text{leak}}^{(1)}+\bar{H}_{\text{leak}}^{(2)})\tau} -I \rVert &\leq \lVert (\bar{H}_{\text{leak}}^{(1)}+\bar{H}_{\text{leak}}^{(2)})\tau\rVert + O\left(\lVert (\bar{H}_{\text{leak}}^{(1)}+\bar{H}_{\text{leak}}^{(2)})\tau \rVert^2\right)\nonumber \\
  &\leq a_1a_3 n^2\tau^2+a_3n\tau^2 +O(n^4\tau^4)\,.
\end{align}

Finally we can bounded $\lVert\ket{\delta\psi_2}\rVert$ by
   \begin{align}
    \lVert\ket{\delta\psi_2}\rVert &\leq \lVert e^{-i(\bar{H}_{I}^{(1)}+\bar{H}_{I}^{(2)}) \tau}-e^{-i(\bar{H}_{\text{noleak}}^{(1)}+\bar{H}_{\text{noleak}}^{(2)})\tau} e^{-i(\bar{H}_{\text{leak}}^{(1)}+\bar{H}_{\text{leak}}^{(2)})\tau} \rVert + \lVert e^{-i(\bar{H}_{\text{leak}}^{(1)}+\bar{H}_{\text{leak}}^{(2)})\tau} -I \rVert \nonumber \\
    &\leq  a_1a_3 n^2\tau^2 + a_1^2a_3n^3\tau^3 + a_3n\tau^2 + O(n^4\tau^4)\,.
  \end{align}

Now we can estimate the total leakage $\lVert\ket{\psi_{\text{leak}}}\rVert$. The leakage during one interval is bounded by $\lVert\ket{\delta\psi_1}\rVert + \lVert\ket{\delta\psi_2}\rVert$. Multiplying the propagators in turn, we obtain the bound for the total leakage:
  \begin{align}
    \lVert\ket{\psi_{\text{leak}}}\rVert &\leq \frac{1}{\tau}\left(\lVert\ket{\delta\psi_1}\rVert + \lVert\ket{\delta\psi_2}\rVert\right) \nonumber \\
    &\leq a_1a_3 n^2\tau + (a_1^2a_3+ \frac{\pi}{8} \Omega_{\text{max}}^3T^3)n^3\tau^2 + O(n^4\tau^3)\,.
  \end{align}
\end{widetext}
If the term $n^2\tau$ is suppressed then every term is suppressed, so a simple choice is to set $\tau=O(n^{-2})$, namely a $O(n^{2})$ cost in $\omega_0$. Then we obtain
\begin{align}
  \sqrt{1-P_{\text{IS}}}&=\lVert\ket{\psi_{\text{leak}}}\rVert \nonumber\\
  &\leq  \frac{\Omega_{\text{max}}T}{8\pi}\left({\Omega_{\text{max}}\Delta_{\text{max}}T^2}+{2\pi KT}\right) \tau_0+O(n^{-1})\,.
\end{align}
which can be arbitrarily small as $n$ grows. The total loss is simply a $O(1)$ cost in $T,\Delta(s)$, and $\Omega(s)$, and a $O(n^{2})$ cost in the blockade strength $\omega_0$. Now if we set $T=T\cdot n^{2}$, $\omega_0=\omega_0 \cdot n^{-2}$, $\Delta(s)=\Delta(s)\cdot n^{-2}$ and $\Omega(s)=\Omega(s)\cdot n^{-2}$, the Schr\"odinger equation Eq.~\eqref{Hamiltonian} remain unchanged, claiming that a $O(n^{2})$ time cost at constant energy cost per atom is sufficient to suppress the leakage. In other words, the time complexity for quantum annealer finding IS is O($n^{2}$). However, please note that the size of the IS in the results is not guaranteed at this point.

\section{Quantum speedup for finding size-2 IS}
If we set $\Delta(s)=0$ and $\Omega(s)=\text{const}$, $H_{\text{eff}}$ is essentially the Hamiltonian of a quantum walk on the median graph. Assuming the graph for the independent set problem has $n$ vertices and $m$ edges. Then it has $(n(n-1)/{2}-m)$ independent sets of size 2. For the quantum walk on the median graph, let $\ket{\psi_0}$ denotes the lattice corresponding to the empty set and $\ket{\psi_{2i}}$ denote the lattices corresponding to  the independent sets of size 2, $i=1,2,\cdots, (n(n-1)/{2}-m)$. After a very short time $s$, the probability of measuring an IS of size 2 is
\begin{align*}
    P_{|\text{IS}|=2}&=\sum_i|\braket{\psi_{2i}|e^{-iH_{\text{eff}}s}|\psi_{0}}|^2 \\
    &=\sum_i\frac{1}{4}s^4|\braket{\psi_{2i}|H_{\text{eff}}^2|\psi_{0}}|^2+O(s^5)\\
    &=\frac{1}{8}\Omega^2T^2s^4\left(\frac{n(n-1)}{2}-m\right)^2+O(s^5)\,.
\end{align*}
For the vast majority of graphs, the number of edges $m$ satisfies $m=\Theta(cn^2)$ with $0\leq c<0.5$(unless the graph is a particularly dense graph). For these graphs and short time evolution, we have $P_{|\text{IS}|=2}\sim s^4n^4$. Thus, the hitting time for IS of size 2 is simply $s\sim n^{-1}$ and it is within the $O(1)$ algorithm runtime. Using the same trick, if we fix the blockade strength to be constant, the hitting time for finding IS of size 2 becomes $n^{-1}\times O(n^{2})=O(n)$, which is faster than the classical $O(n^2)$ algorithm and reproduces the Grover type quantum speedup.

For extremely dense graphs ($m=n^2/2+o(n^2)$), let us assume $n(n-1)/2-m= \beta n^{\gamma}+o(n^{\gamma})$ with $0<\gamma<2$ and $\beta>0$, the hitting time of the quantum algorithm becomes $O(n^{2-\gamma/2})$, still faster than the classical $O(n^2)$ runtime.

\section{The result obtained by two well-known adiabatic conditions}
Let's first introduce two commonly used adiabatic conditions. One well known adiabatic condition~\cite{trivialbound} is $t_f\gg \max_{s\in [0,1]}|\braket{\varepsilon_0(s)|\partial_sH(s)|\varepsilon_1(s)}|/{|\varepsilon_1(s)-\varepsilon_0(s)|^2}$,where the dynamics is driven by $i\md \ket{\psi(s)}/\md s=t_f H(s)\ket{\psi(s)}$, $s\in[0,1] $ and $\ket{\varepsilon_i(s)}, j\in\{0,1,2\cdots\}$ denote the instantaneous eigenstates of $H(s)$, and the initial state is prepared in $\ket{\psi(0)}=\ket{\varepsilon_0(0)}$. This condition is regarded as the approximate version of adiabatic theorem, because it cannot provide either strict inequalities or bounds on the closeness between the actual time-evolved state and the desired eigenstate.

The second condition, provided by Jansen, Ruskai and Seiler (JRS)~\cite{JRS}, is regarded as the rigorous version of the adiabatic theorem  which has been used to prove the speedup of many quantum adiabatic algorithms. Suppose that the spectrum of $H(s)$ restricted to $P(s)$ consists of $m(s)$ eigenvalues (each possibly degenerate with crossing permitted) separated by a gap of $\Delta_E(s)$ from the rest of the spectrum of $H(s)$, a sufficient adiabatic condition is $t_f\gg\max\{{m(s)\lVert H^{[2]}(s)\rVert}/{\Delta_E^2(s)},\\{m^{1.5}(s)\lVert H^{[1]}(s)\rVert^2}/{\Delta_E^3(s)},{m(s)\lVert H^{[1]}(s)\rVert}/{\Delta_E^2(s)}\}$, where $H^{[1]}$ and $H^{[2]}$ denote the first or second time derivatives of $H(s)$.

We will carry out an analogous interaction transformation so that we can apply these two adiabatic conditions properly. If we set $U_I(s) = \mathcal{T}e^{-iT\int_0^s(\frac{\Omega(s)}{2}\sum_j \sigma_j^x
-\Delta(s)\sum_j {n}_j)\mathrm{d}s'}$ and $\ket{\Phi_I(s)}=U_I^\dagger(s)\ket{\Phi_S(s)}$, again we obtain an interaction picture transformation with ${H}_{I}(s)=U_I^\dagger(s)(\omega_0 T \sum_{\langle i,j\rangle}{n}_i{n}_j) U_I(s)$, whose instantaneous eigenstates are denoted by $\ket{\psi_i(s)}=U_I^{\dagger}(s)\ket{\phi_i}$. 

${H}_{I}(s)$ is simply a rotation of $\omega_0 T \sum_{\langle i,j\rangle}{n}_i{n}_j$ and naturally, $\ket{\Phi_S(s)}$ is restricted to the IS subspace iff $\ket{\Phi_I(s)}$ is in the instantaneous ground states of ${H}_{I}(s)$. Thus, the leakage out of the IS subspace is equivalent to the ground state adiabatic condition of ${H}_{I}(s)$ and now we can use the aforementioned adiabatic conditions to study the leakage. Applying the approximate version, it gives simply $\max_{s\in [0,1]}\Omega(s)/2\omega_0\ll 1$. This condition provides no restriction for $\Delta(s)$ and the quantum state will be frozen in the initial state (the empty set) as $\Omega\to 0$. Also, there is no quantitative standard for ``$\ll$'', e.g. we find experiments with $\max_{s\in [0,1]}\Omega(s)/2\omega_0\approx 0.1$ still suffer from the IS subspace leakage~\cite{Lukin2022}, which can be explained by Eq. \eqref{hopping}: large detuning will resonate with blockade, and lead to the leakage.

Applying the rigorous JRS condition, we note that the subspace degeneracy $m$ is the number of independent sets and $m=O(a^n)$ according to Ramsey theory~\cite{Ramsey} in King's graph. This leads to an exponential upperbound for running time $T$. Both of the usual adiabatic conditions give disappointing upperbounds.

\section{The classical minimum-degree greedy algorithm}
\begin{figure}[H]
  \centering
  \begin{minipage}{\columnwidth}
  \begin{algorithmic}[1]
    \Require Undirected graph $G=(V,E)$
    \Ensure Independent set $I$
    \State $I \gets \varnothing$, \ $U \gets V$
    \While{$U \neq \varnothing$}
      \State choose $v \in \arg\min_{u \in U} \deg_{G[U]}(u)$ \Comment{tie-break: smallest label}
      \State $I \gets I \cup \{v\}$
      \State $U \gets U \setminus \big(\{v\} \cup N_{G[U]}(v)\big)$
    \EndWhile
    \State \Return $I$
  \end{algorithmic}
  \end{minipage}
  \caption{Minimum-degree greedy for MIS.}
  \label{fig:mdg-mis}
\end{figure}

The complexity of this algorithm is $O(|V|+|E|)$, which is $O(n^2)$ for the Erd\H{o}s--R\'enyi graph.

\section{Numerical simulation of the quantum algorithm}

For the results in Fig. \ref{fig:4}, namely the quantum $O(n^2)$ algorithm for finding large ISs, we sampled 100 graphs for each data point. And we set $T=100$, $\omega_0=n^2$, $\Omega(s)= \sin(\pi s)$ and $\Delta(s)= \cos(\pi s)$. The tested graphs are the Erd\H{o}s--R\'enyi graph with $p=0.8$. 

For the results in Fig. \ref{fig:3}, namely the quantum $O(n)$ algorithm for finding size-2 ISs, we sampled 100 graphs for each data point. And we set $T=20/n$, $\omega_0=n^2$, $\Omega(s)= 1$ and $\Delta(s)= 0$. The tested graphs are the Erd\H{o}s--R\'enyi graph with $p=0.5$.

\end{document}